 \documentclass[12pt, draftclsnofoot, onecolumn]{IEEEtran}

\ifCLASSINFOpdf
  \usepackage[pdftex]{graphicx}
  \graphicspath{{../pdf/}{../jpeg/}}
  \DeclareGraphicsExtensions{.pdf,.jpeg,.png}
\else
  \usepackage[dvips]{graphicx}
  \graphicspath{{../eps/}}
  \DeclareGraphicsExtensions{.eps}
\fi
\usepackage[cmex10]{amsmath}
\usepackage{amsfonts}
\usepackage{algorithm, algorithmic}
\usepackage{array}
\usepackage{mdwmath}
\usepackage{mdwtab}
\usepackage{eqparbox}
\usepackage{stfloats} 
\usepackage{amsthm}
\usepackage{amsmath,bm}
\usepackage[colorlinks = true, citecolor = {blue} ]{hyperref}

\newtheorem{theorem}{Theorem}
\newtheorem{theorem1}{Theorem}

\newtheorem{corollary}[theorem1]{Corollary}

\begin{document}

\title{Spectrum Refarming: A New Paradigm of Spectrum Sharing for Cellular Networks}

\author{Shiying~Han,~\IEEEmembership{Student Member,~IEEE,}~Ying-Chang Liang,~\IEEEmembership{Fellow,~IEEE,}\\and~Boon-Hee~Soong,~\IEEEmembership{Senior Member,~IEEE}
\thanks{S. Han and B.-H. Soong are with the School of Electrical and Electronic Engineering, Nanyang Technological University, Singapore, 639798 (e-mail: shan1@e.ntu.edu.sg, ebhsoong@ntu.edu.sg).}%
\thanks{Y.-C. Liang is with the Institute for Infocomm Research, A*STAR, 1 Fusionopolis Way, Connexis, Singapore 138632 (e-mail: ycliang@i2r.a-star.edu.sg).}
\thanks{Part of this work was submitted to IEEE Globecom 2014 in Austin, TX USA \cite{IEEEhowto: Han}.}
}



\maketitle
\IEEEpeerreviewmaketitle

\begin{abstract}
Spectrum refarming (SR) refers to a radio resource management technique which allows different generations of cellular networks to operate in the same radio spectrum. In this paper, an underlay SR model is proposed, in which an Orthogonal Frequency Division Multiple Access (OFDMA) system refarms the spectrum of a Code Division Multiple Access (CDMA) system through intelligently exploiting the interference margin provided by the CDMA system.
We investigate the mutual effect of the two systems by evaluating the asymptotic signal-to-interference-plus-noise ratio (SINR) of the users, based on which the interference margin tolerable by the CDMA system is determined. By using the interference margin together with the transmit power constraints, the uplink resource allocation problem of OFDMA system is formulated and solved through dual decomposition method. Simulation results have verified our theoretical analysis, and validated the effectiveness of the proposed resource allocation algorithm and its capability to protect the legacy CDMA users.
The proposed SR system requires the least information flow from the CDMA system to the OFDMA system, and importantly, no upgrading of legacy CDMA system is needed; thus it can be deployed by telecom operators to maximize the spectral efficiency of their cellular networks.

\end{abstract}


\section{Introduction}

\IEEEPARstart{T}{o} meet the increasing demand on high capacity, high data rate, and low latency, in the past two decades, the cellular networks have evolved from the 2nd generation (2G) network, such as Global System for Mobile Communications (GSM), to the 3rd generation (3G) network, such as wideband code division multiple access (CDMA) and CDMA-2000, and to the 4th generation (4G) network, such as long-term evolution (LTE), and LTE-advanced. When a new generation (i.e., 4G) network is introduced, the mobile traffic will gradually migrate from the older generation (i.e., 2G/3G) networks towards the new one \cite{IEEEhowto: Cisco}. Fig.~\ref{fig.1} illustrates the evolution of subscription as cellular technology evolves from one generation to another \cite{IEEEhowto: Pysavy}. As can be seen, the legacy networks will not cease operations immediately after the new one is deployed, and in fact they will have to continue to provide services to the legacy users for a significant period of time before being phased out. During the transition period, the spectrum assigned to the legacy networks will experience low spectral efficiency when the number of legacy users is lower than the designed network capacity.

Spectrum refarming (SR) is an innovative spectrum sharing technique which allows different generations of cellular networks to operate in the same radio spectrum. By doing so, the spectral efficiency of cellular networks can be greatly improved. Since radio spectrum is a limited and expensive resource, SR is considered as a promising solution for mobile service operators not only to provide cost effective services to their customers, but also to solve the spectrum scarcity problem faced by them.
Recently, Lin et al. \cite{IEEEhowto: Lin1} proposed an LTE/GSM SR system for an Orthogonal Frequency Division Multiple Access (OFDMA) system to refarm the GSM band by utilizing the subbands that are not occupied by the GSM system.

\begin{figure}
    \centering
    \includegraphics[width=7.5cm,height=6.3cm]{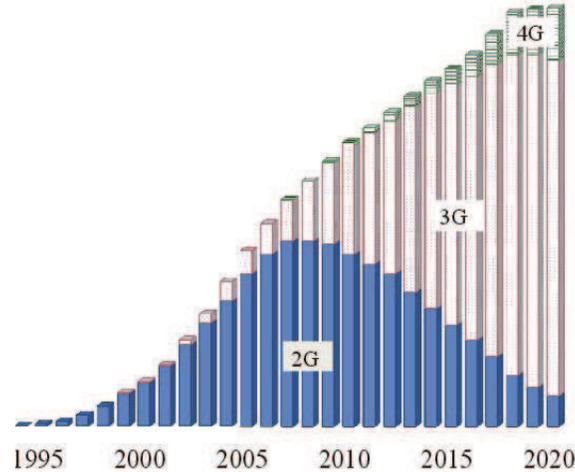}
    \caption{Subscriptions evolution as cellular network technology evolves (Source: CDG, ZTE and \cite{IEEEhowto: Pysavy}).} \label{fig.1}
\end{figure}

Conventionally, spectrum sharing is implemented in two ways \cite{IEEEhowto: Liang}: overlay spectrum sharing, which allows the secondary users opportunistically access the unused spectrum of the legacy (primary) users, and underlay spectrum sharing, which allows the secondary and primary users co-transmit at the same band. Similarly, there are two types of SR models: overlay SR model and underlay SR model. The LTE/GSM SR system \cite{IEEEhowto: Lin1} operating in GSM band belongs to an overlay SR model.

In this paper, a new SR model is proposed for an OFDMA system to refarm the spectrum assigned to a Code Division Multiple Access (CDMA) system through intelligently exploiting the interference margin provided by the CDMA system. Due to the wide band nature of the CDMA and OFDMA systems, the OFDMA/CDMA SR system will operate in an underlay manner, thus the system belongs to the underlay SR model.

The operating principle of OFDMA/CDMA SR system is as follows. Let us consider a single-cell direct sequence CDMA uplink, in which the CDMA users are assigned with random spreading codes. When the receiver is applied, such as match filter (MF) or linear minimum mean-square-error (MMSE) receiver, there exists an inter-user interference, which is related to the CDMA system load defined as the ratio of total number of CDMA users to the spreading gain \cite{IEEEhowto: Verdu}. When the processing gain and number of users are both large, the signal-to-interference-plus-noise-ratio (SINR) of the receiver output will converge to a limiting SINR, which is independent of the specific spreading codes, thanks to the random matrix theory. Thus, for a given receive power, the maximum CDMA system load depends on the target SINR only. In another word, when the CDMA system is operating at a lower load, the CDMA users will experience less interference. Thus, there exists an interference margin that can be tolerated by CDMA users, with which the same target SINR for CDMA users can still be maintained. This interference margin in fact defines the maximum interference power that can be introduced by the OFDMA system to the CDMA system for the OFDMA/CDMA SR system.

To implement a successful OFDMA/CDMA SR system, we identify and address several key challenges as follows.

First and foremost, the interference margin provided by the CDMA system has to be predicted. In conventional spectrum sharing studies, the primary and secondary systems usually operate with the same access scheme, and the interference margin is given as a predefined threshold \cite{IEEEhowto: Ghasemi}-\cite{IEEEhowto: Kang}. In our SR system, there exist different access schemes, thus we have to carefully quantify the interference margin based on the SINR analysis of CDMA system. Although there are extensive investigations on the CDMA SINR in single/multi-cell \cite{IEEEhowto: Tse}, \cite{IEEEhowto: Viswanath}, flat/frequency-selective fading \cite{IEEEhowto: Verdu}, \cite{IEEEhowto: Miller} scenarios, few effort has been put on the SINR performance when the CDMA users are interfered by a system with a different access scheme.
Furthermore, when evaluating the SINR of CDMA users, it is usually required to know the specific spreading codes adopted by CDMA users and the channel state information (CSI) between CDMA users and the base station (BS).
This could be impractical for the SR system, due to limited information exchange between the two systems. Motivated by \cite{IEEEhowto: Verdu}, \cite{IEEEhowto: Tse} that the CDMA SINR converges in probability to a deterministic value when the system dimensions get large, we resort to the asymptotic analysis to predict the interference margin for our SR system. Thus, several questions need to be answered.
How can the asymptotic SINR of CDMA user accommodate the interference from the OFDMA system, since the interference here is colored in spectrum? When the CDMA system adopts different receive filters, such as MF and MMSE filter, is the effect of OFDMA interference identical? What will be the impact of the CDMA system to the OFDMA system?
Answering these questions in this work will reveal the essence of the OFDMA/CDMA SR system.

Second, once the CDMA interference margin is obtained, this margin should be effectively used by the OFDMA system to perform resource allocation. Again, in conventional spectrum sharing studies, to keep the interference to primary receiver below the interference margin, the CSI of the channel from secondary transmitter to primary receiver, i.e., cross-channel, has to be known \cite{IEEEhowto: Ghasemi}-\cite{IEEEhowto: Kang}, \cite{IEEEhowto: Almalfouh}, \cite{IEEEhowto: Khoshkaholgh1} or partially known \cite{IEEEhowto: Soltani}. In practice, however, the cross-channel CSI is usually difficult to be obtained due to the non-cooperation between primary and secondary systems.


To solve the above problems, we first establish the signal model of the proposed OFDMA/CDMA SR system. In our model, both OFDMA and CDMA system share the same cell site and same BS antenna. This sharing could be possible as both CDMA and OFDMA networks belong to the same operator. Thus, the CSI of cross-channel is known, as it equals to the CSI of secondary signal channel itself. This is an interesting and fundamental model for the proposed SR system, from which various studies could be made. We then quantify the asymptotic SINR performance of the CDMA and OFDMA users in the proposed SR system using the random matrix theory and the large number law, from which the interference margin tolerable by CDMA system is derived.
When evaluating the SINR of OFDMA users, the interference from CDMA system is also taken into consideration.
Finally, by using the interference margin together with the transmit power constraints, the resource allocation problem of the OFDMA system is formulated and solved through dual decomposition method.

It is pointed out that in \cite{IEEEhowto: Almalfouh}, \cite{IEEEhowto: Khoshkaholgh1}, and \cite{IEEEhowto: Soltani}, where the spectrum is shared by the CDMA and OFDMA systems, the interference margin was a predefined value, and there is no justification on how to determine this value. Furthermore, the interference from primary transmitter to secondary receiver was not taken into consideration.


The remainder of the paper is organized as follows.
In Section II, the proposed OFDMA/CDMA SR system, and its signal transmission model are presented.
In Section III, the signal detection for CDMA and OFDMA users is investigated.
The asymptotic SINR of CDMA users with different CDMA receive filters are analyzed, and the interference margin that can be exploited by OFDMA system is determined.
After that, the resource allocation problem of OFDMA system is formulated and solved in Section IV.
Simulation results that validate our analysis is presented in Section V.
Finally, the paper is concluded in Section VI.

The notations used in this paper are as follows.
The boldface upper case letters denote matrices, and the boldface lower case letters denote vectors.
$\mathbb E[\cdot]$ denotes expectation.
We use superscripts $(\cdot)^H$ and $(\cdot)^T$ to denote the conjugate transpose and transpose of a matrix or a vector.
$\bm I$ and $\rm{diag(\cdot)}$ denote the identity matrix and diagonal matrix, respectively.
$\rm{Tr\{\cdot\}}$, $|\cdot|$, and $\|\cdot\|$ denote the trace of a matrix, absolute value, and Euclidean distance, respectively.

\section{System Model}

The uplink of OFDMA/CDMA SR system is considered in this paper.
The signal transmission model from user terminals to the BS is illustrated in Fig. \ref{fig.2}.
Here, we assume that single antenna is deployed at the BS and all user terminals.
Furthermore, both systems share the same cell site and same BS antenna.
This sharing is reasonable, because the operator can adopt SR technique by adding OFDMA transceiver to the CDMA cell cite.
The signal received by the common antenna can be passed to the CDMA module and the OFDMA module for their respective signal detection. It is pointed out that the results obtained in this paper can also be extended to the case that different antennas are used at the BSs of the two systems.

\subsection{SR System Description}
\begin{figure}
    \centering
    \includegraphics[height=4.8cm]{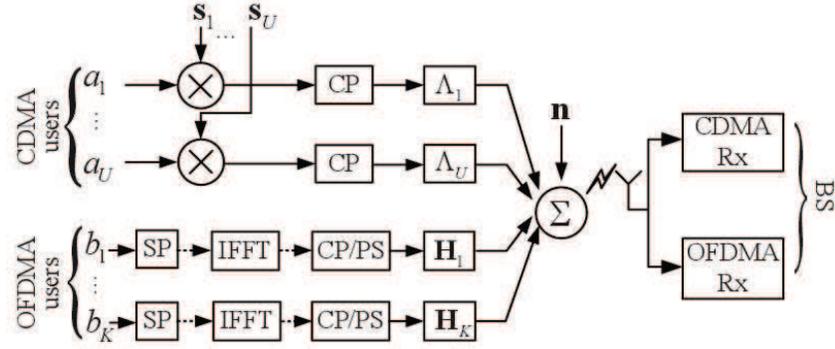}
    \caption{The system model of uplink OFDMA/CDMA SR system.} \label{fig.2}
\end{figure}

Denote $W$ as the effective bandwidth licensed to the CDMA system operation, $U$ the total number of CDMA users.
Each CDMA user is assigned with a random spreading code with spreading gain $N$.
Thus the chip duration and symbol duration for the CDMA system are denoted as
$T_c=\frac{1}{W}$, and $T_s=\frac{N}{W}$, respectively.
The OFDMA system operates in the same spectrum as the CDMA system.
Let $K$ be the total number of OFDMA users.
Suppose the FFT size is $N$, and the OFDMA symbol duration is chosen to be equal to the CDMA symbol duration $T_s$.
The cyclic-prefix (CP) is used to take care of the multipath effect of the wireless channel.
The overall bandwidth $W$ is thus split into $N$ equally-spaced orthogonal subcarriers.

Take the wideband CDMA uplink as an example.
In practice, it operates at a 5MHz bandwidth with the chip rate of 3.84Mcps.
The spreading gain can vary from 2 to 256 \cite{IEEEhowto: 3GPP}.
The LTE can adopt 256 subcarriers when working at 5MHz mode with subcarrier spacing of 15kHz, The sampling rate is thus $15\rm{KHz}\times 256 = 3.84\rm{MHz}$ that equals to the WCDMA chip rate.
Thus, the two systems can easily get synchronized with the same clock reference.

All the wireless links are quantified by distance-dependent path loss, large-scale shadowing, and small-scale fading.
For simplicity, the distance-dependent path loss and large-scale shadowing are treated to be 1, and the average power of small-scale fading is also normalized to be 1.
With this normalization, the transmission power in fact defines the average receive power at the BS. Furthermore, the small-scale fading is modeled with $L$ equally-spaced multipaths, and the time delay between two consecutive paths is $T_{c}$. Moreover, the channels are assumed to be invariant within each symbol duration. Next, we will look into the details of signal transmissions in the SR system.

\subsection{CDMA Signal}

As shown in Fig. \ref{fig.2}, let $a_u$ be the data symbol to be transmitted by CDMA user $u$.
The symbol $a_{u}$ is spread by spreading code ${\bm{s}}_u$, where ${{\bm{s}}_u} = {\left[ {{s_{u,1}},{s_{u,2}},\dots,{s_{u,N}}} \right]^T}$ with unit power, i.e.,  ${{\mathbb E} [\| {{{\bm{s}}_u}} \|^2]} = 1$, and covariance of ${\mathbb E}\left[ {{{\bm{s}}_u}{\bm{s}}_u^H} \right] = \frac{1}{N}{\bm{I}}$.
The spreading codes ${\bm s}_1,\dots, {\bm s}_{U}$ for different users are chosen to be independent with each other.
With perfect power control \cite{IEEEhowto: Viterbi}, we assume that all users have equal average receive power $q$, i.e., ${\mathbb E}[|a_1|^2] = \cdots = {\mathbb E}[|a_U|^2] =q$.

For ease of presentation, we assume the CDMA system also has CP, similar to the OFDMA system, though the results derived in the paper can be readily extended to normal CDMA systems\footnote[1]{The effect of ISI in normal CDMA systems without CP insertion can be negligible, when the number of the interfered chips is small compared to the spreading gain. This is true especially when $N\to \infty$.  In the normal CDMA system with limited spreading gain, the ISI can be largely mitigated by post-processing, such as the overlap-save scheme \cite{IEEEhowto: Mahmoud}.}.
For every $N$ chips coming from the same symbol, a CP containing $G$ ($G \ge L-1$) chips are inserted before the chip signals are transmitted over the wireless channel.
The insertion of CP at transmitter and removal of CP at receiver avoid the inter-symbol-interference (ISI) between successive symbols caused by multipath effect.

The time-domain CDMA received signal at the BS after CP removing is given by
\begin{align}  \label{2.1}
	{{\bm{r}}_c} = \sum\limits_{u=1}^U {{{\bm{C}}_u}{{\bm{s}}_u}{a_u}},
\end{align}
where ${{\bm{r}}_c}$ is $N \times 1 $ vector, ${\bm{C}}_u$ is $N \times N$ matrix
\begin{align} \label{2.2}
	{{\bm{C}}_u} = \left( {\begin{array}{*{20}{c}}
{{h_{u,L}}}\\0\\ \vdots \\0
\end{array}\begin{array}{*{20}{c}}
{{h_{u,L - 1}}}\\{{h_{u,L}}}\\ \vdots \\0
\end{array}\begin{array}{*{20}{c}}
 \cdots \\ \cdots \\ \ddots \\ \cdots
\end{array}\begin{array}{*{20}{c}}
{{h_{u,1}}}\\0\\ \vdots \\{{h_{u,L}}}
\end{array}\begin{array}{*{20}{c}}
 \cdots \\ \cdots \\ \ddots \\
 \cdots
\end{array}\begin{array}{*{20}{c}}
0\\0\\ \vdots \\
{{h_{u,1}}}
\end{array}} \right),
\end{align}
whose entry $h_{u,l}$ is the impulse response of the $l$th path of user $u$. Because ${\bm{C}}_u$ is circulant, it can be factorized as
\begin{align} \label{2.3}
    {{\bm{C}}_u} = {{\bm{W}}^H}{{\bm{\Lambda}} _{u}}{\bm{W}},
\end{align}
where ${{\bm{W}}}$ is the $N$-point Fourier transform matrix whose entries are $\frac{1}{{\sqrt N }}{( {{e^{-\frac{{2\pi i}}{N}ab}}})_{a,b = 0,\dots,N-1}}$, and ${\bm{\Lambda}}_{u}={\rm{diag}}(\lambda_{u,1},\dots,\lambda_{u,N})$ contains the frequency-domain channel response of user $u$.
Substituting \eqref{2.3} in \eqref{2.1} yields
\begin{align} \label{2.4}
    {{\bm{r}}_c} = \sum\limits_{u=1}^U {{{\bm{W}}^H}{{\bm{\Lambda}} _{u}}{\bm{W}}{{\bm{s}}_u}{a_u}}.
\end{align}

\subsection{OFDMA Signal}

According to Fig. \ref{fig.2}, the $N$-dimensional OFDMA received signal can be written as 
\begin{align} \label{2.5}
	{{\bm{r}}_o} = \sum\limits_{k=1}^K {{{\bm{W}}^H}{{\bm{H}}_{k}}{{\bm{b}}_k}} = {{\bm{W}}^H}{{\bm{H}}}{\bm{b}}.
\end{align}
Here, $ {{{\bm{b}}_k}}  = {[ {{b_{k,1}},{b_{k,2}},\dots,{b_{k,N}}} ]^T}$ is the transmission data vector for each user $k=1,\dots,K$, and ${\mathbb E}[ {{{| {{b_{k,n}}}|}^2}} ] = {p_{k,n}}$ is the transmission power of user $k$ on subcarrier $n$.
${{\bm{H}}_k} = {\rm{diag}}\left( {{H_{k,1}},{H_{k,2}},\dots,{H_{k,N}}} \right)$ is frequency-domain channel response matrix for user $k$. Moreover, $ {\bm{b}} = \sum\nolimits_{k=1}^K {{{\bm{b}}_k}} $ and
${\bm{H}} = {\rm{diag}} \left( {{H_{k[1],1}},{H_{k[2],2}},\dots,{H_{k[N],N}}} \right)$, where ${\left\{ {k[n]} \right\}_{n=1,\dots,N}}$ are the user index, indicating the user that is transmitting over subcarrier $n$.

\subsection{Compound Received Signal at BS}

Assume that both systems share the same receive antenna at the BS, and the OFDMA symbols are synchronously received at the BS with the CDMA symbols\footnote[2]{The OFDMA system can adopt a timing advance mechanism that is widely used in 2G/3G cellular networks \cite{IEEEhowto: Zhang}. As OFDMA transmits, it first synchronizes the frame with CDMA system. Then they will keep synchronous as they have the same symbol duration, and the use of CP can tolerate certain amount timing mismatch.}. The receive signal ${\bm{r}}$ is thus composed by ${\bm{r}}_c$, ${\bm{r}}_o$ and the additive white Gaussian noise ${\bm{n}}$, where ${\bm{n}} \sim {\cal{CN}}\left( {0,{\sigma^2}{\bm{I}}} \right)$, which is represented as
\begin{align} \label{2.6}
	{\bm{r}} = \sum\limits_{u=1}^U {{{\bm{W}}^H}{{\bm{\Lambda}} _{u}}{\bm{W}}{{\bm{s}}_u}{a_u}}  + {{\bm{W}}^H}{{\bm{H}}}{\bm{b}} + {\bm{n}}.
\end{align}
The corresponding frequency-domain compound signal is
\begin{align} \label{2.7}
	\tilde {\bm{r}} = \sum\limits_{u =1}^U {{{\bm{\Lambda}} _{u}}{{\tilde {\bm{s}}}_u}{a_u}}  + {{\bm{H}} }{\bm{b}} + \tilde {\bm{n}},
\end{align}
where $\tilde {\bm{r}}={\bf{W}} {\bm{r}}$, $\tilde {\bm{s}}_u = {\bm W}{\bm{s}}_u$, and $\tilde {\bm{n}} = {\bm{W}} {\bm{n}}$.

It can be seen from \eqref{2.7} that when the OFDMA and CDMA uplink transmissions co-occur in the system, they will interfere with each other at the BS. In the following sections, we will describe the respective signal detections for the two systems, and quantify their mutual interference, which will then be used to design the resource allocation scheme for the OFDMA system.

\section{Disjoint Detection of SR System}

Without cooperation between the two systems, disjoint detection is applied to CDMA receiver and OFDMA receiver which abstract their desired data by treating the remained part as interference plus noise.
In practice, linear receivers are adopted by CDMA system due to their complexity advantage.
Denoting ${{{\bm{e}}_u}}$ as the MF receiver and ${{{\bm{d}}_u}}$ as the MMSE receiver of CDMA user $u, \forall u=1,\dots,U$, in the following subsections, we will derive the SINR expressions for the CDMA and OFDMA users when disjoint detection is applied.

\subsection{CDMA SINR with MF Receiver}

Taking the CDMA user $u$ as an example, the MF is designed as ${\bm{e}_u} = {{\bm{\Lambda}} _{u}}{\tilde {\bm{s}}_u}$, and the MF output, ${{\hat a}_u}^{({\rm{MF}})} = {{\bm{e}}_u^H}\tilde {\bm{r}} $, can be expressed as
\begin{align} \label{3.1}
	{{\hat a}_u}^{({\rm{MF}})} = {\bm{e}}_u^H{{\bm{e}}_u}{a_u} + {\bm{e}}_u^H\left( {\sum\limits_{i=1,i\not= u}^U {{{\bm{\Lambda}} _i}} {{\tilde {\bm{s}}}_i}{a_i} + {\bm{Hb}} + \tilde {\bm{n}}} \right).
\end{align}
The first term in \eqref{3.1} is the desired signal, and the remainder is interference plus noise. The SINR can be calculated as
\begin{align} \label{3.2}
	\gamma _u^{a({\rm{MF}})} = \frac{{{q}{{\left| {{\bm{e}}_u^H{{\bm{e}}_u}} \right|}^2}}}{{{\bm{e}}_u^H\left( {\sum\nolimits_{i =1, i\not= u}^U {{q}{{\bm{e}}_i}{\bm{e}}_i^H}  + {\bm{\Sigma}}  + {\sigma ^2}{\bm{I}}} \right){{\bm{e}}_u}}}.
\end{align}
The matrix ${\bm{\Sigma}}$ is the $N\times N$ covariance matrix of $\bm{Hb}$ with $n$th diagonal entry being $\sigma_n^2$.
By denoting $g_{k,n} = |H_{k,n}|^2$ as the channel gain of OFDMA user $k$ on subcarrier $n$, $\sigma_n^2$ can be expressed as $\sigma_n^2 = \sum_{k=1}^K p_{k,n}g_{k,n}$, which represents the interference power seen by CDMA on subcarrier $n$.

Eq. \eqref{3.2} shows that the SINR of a CDMA user depends on the specific spreading codes as well as the CSI of all CDMA users in the system.
With the spectrum sharing between the CDMA and OFDMA system, the SINR of CDMA user is further affected by the transmission power and CSI of OFDMA system.
Without knowing these information, it is difficult to exactly calculate the SINR of the CDMA users.
Fortunately, it has been proven that when $\bm{\Sigma} = \bm 0$, \eqref{3.2} converges in probability to a deterministic value when the dimensions, $(N, U)$, of the CDMA system become large, which means the asymptotic SINR of pure CDMA system is independent of the specific spreading codes.
This inspires us to investigate the asymptotic SINR of CDMA users when OFDMA system co-exists.

{\bf{AS1}}: We consider a large CDMA system, in which $N \to \infty , U \to \infty,$ but $\frac{U}{N}$ converges to a constant parameter $\alpha$, which represents the CDMA system load.

\begin{corollary} \label{corollary: MF-SINR}
With {\bf{AS1}}, the SINR of CDMA user $u$, $\gamma^{a\rm{(MF)}}_u$, in the SR system converges in probability to
\begin{align} \label{3.3}
	\tilde \gamma _u^{a({\rm{MF}})} =
	\frac{{{q}{{\left( {\frac{1}{N}\sum\nolimits_{n=1}^N {{{\left| {{\lambda _{u,n}}} \right|}^2}} } \right)}^2}}}{{\frac{1}{{{N^2}}}\sum\nolimits_{n=1}^N {\left( {{{\left| {{\lambda _{u,n}}} \right|}^2}\sum\nolimits_{i=1,i\not=u}^U {{q}{{\left| {{\lambda _{i,n}}} \right|}^2}} } \right)}  + \frac{1}{N}\sum\nolimits_{n=1}^N {( {{{\left| {{\lambda _{u,n}}} \right|}^2}\sigma _n^2} )}  + \frac{1}{N}\sum\nolimits_{n=1}^N {{{\left| {{\lambda _{u,n}}} \right|}^2}} {\sigma ^2}}},
\end{align}
where $\sigma _n^2 = \sum\nolimits_{k=1}^K {{p_{k,n}}{g_{k,n}}}, \forall n=1,\dots,N $.
\end{corollary}


\begin{proof}
To derive the asymptotic value of $\gamma_u^{a\rm{(MF)}}$, we can derive asymptotic values of the nominator and denominator in \eqref{3.2} separately.
First, the nominator in \eqref{3.2} is given by
\begin{align} \label{c1.1}
    q{\left| {\bm{e}_u^H{\bm{e}_u}} \right|^2} = q{\left| {\tilde{\bm s}_u^H \bm{\Lambda} _u^H{\bm{\Lambda} _u}{\tilde{\bm {s}}_u}} \right|^2} = q{\left| {{\rm{Tr}}\left\{ {\tilde{\bm s}_u^H\bm{\Lambda} _u^H{\bm{\Lambda} _u}{{\tilde{\bm s}}_u}} \right\}} \right|^2}.
\end{align}
As $N\to\infty$, according to \cite{IEEEhowto: Verdu}$, {\rm{Tr}}\left\{ {\tilde{\bm s}_u^H \bm{\Lambda} _u^H{\bm{\Lambda} _u}{{\tilde {\bm s}}_u}} \right\} \mathop  \to \limits^{a.s.} $ $ {\frac{1}{N}{\rm{Tr}} \left\{ {\bm{\Lambda} _u^H{\bm{\Lambda} _u}} \right\}} = \frac{1}{N}\sum\nolimits_{n = 1}^N {{{\left| {{{\lambda} _{u,n}}} \right|}^2}}$, \eqref{c1.1} is thus given by
\begin{align}
    q{\left| {\bm{e}_u^H{\bm{e}_u}} \right|^2} \mathop \to \limits^{a.s.} q{\left( {\frac{1}{N}\sum\nolimits_{n = 1}^N {{{\left| {{\lambda _{u,n}}} \right|}^2}} } \right)^2}.
\end{align}
Similarly, the asymptotic value of the denominator in \eqref{3.2} can be derived as
\begin{align} \label{3.4}
    & {{\bm{e}}_u^H\left( {\sum\nolimits_{i =1, i\not= u}^U {{q}{{\bm{e}}_i}{\bm{e}}_i^H}  + {\bm{\Sigma}}  + {\sigma ^2}{\bm{I}}}  \right){{\bm{e}}_u}} \notag \\
    & \mathop \to \limits^{a.s.}
    \frac{1}{N}{\rm{Tr}} \left\{{{{\bm{\Lambda}} _u} {\bm{\Lambda}} _u^H\left( {\sum\nolimits_{i=1, i\not= u}^U {{q}{{\bm{e}}_i}{\bm{e}}_i^H}  + {\bm{\Sigma}}  + {\sigma ^2}{\bm{I}}} \right)} \right\},
\end{align}
which can be further derived to the denominator of \eqref{3.3}. Then the corollary is proven.
\end{proof}

Corollary 1 offers us the following insights.
Besides the inter-user-interference due to the independence and randomness of the spreading codes used by CDMA users, the OFDMA introduces interference on each subcarrier, acting as colored noise across the overall spectrum due to the non-uniform power allocation among subcarriers and channel fading selectivity, which is whitened at the output of the MF.

Corollary 1 describes the asymptotic SINR for CDMA users with frequency selective fading channels. If the CDMA users are with flat fading channels, i.e.,  ${\bm{\Lambda}}_{u}=\lambda_u {\bm{I}}$, (10) can be simplified as
\begin{align} \label{3.5}
    \tilde \gamma _u^{a({\rm{MF}})} = \frac{{{q}{{\left| {{\lambda _u}} \right|}^2}}}{{\frac{1}{N}\sum\nolimits_{i=1, i\ne u}^U {{q}{{\left| {{\lambda _i}} \right|}^2}}  + \frac{1}{N}\sum\nolimits_{n=1}^N { \sum\nolimits_{k=1}^K {{g_{k,n}}{{p_{k,n}}} } }  + {\sigma ^2}}}.
\end{align}
If the CDMA users are with AWGN channels, i.e., ${\lambda _u}=1$, \eqref{3.5} can be further simplified to
\begin{align} \label{3.6}
    \tilde \gamma _u^{a({\rm{MF}})} = \frac{q}{{\alpha q + \frac{1}{N}\sum\nolimits_{n=1}^N {\sum\nolimits_{k=1}^K g_{k,n} {{p_{k,n}}} }  + {\sigma ^2}}}.
\end{align}

Eq. \eqref{3.3} still cannot be used by OFDMA system to predict the interference margin, since the SINR in \eqref{3.3} depends on the CSI of CDMA users, which is not available to OFDMA system in practice.
Next, we will solve this problem by looking at the case when the CDMA users are with rich multipath components. We make the following additional assumption:

{\bf{AS2}}: The number of multipaths $L$ for CDMA users is large, and the wireless channel follows uniform power delay profile \cite{IEEEhowto: Goldsmith}. That is, each channel tap follows complex Gaussian distribution, $h_{u,l}\sim {\cal{CN}} (0,1/L)$.

\begin{corollary} \label{corollary: MF-SINR}
    With {\bf{AS1}} and {\bf{AS2}}, the asymptotic SINR of all the CDMA users in \eqref{3.3} converges to
    \begin{align}\label{3.7}
        \gamma^{a({\rm{MF}})} = \frac{{{q}}}{{\alpha  q  + \overline{{\sigma _n^2}} + {\sigma ^2}}},
    \end{align}
    where $\overline{{\sigma _n^2}} = \frac{1}{N} \sum_{n=1}^{N} \sum_{k=1}^{K} p_{k,n} g_{k,n}$.
\end{corollary}

\begin{proof}
The proof is given in Appendix A.
\end{proof}
The limiting SINR for CDMA users shown in \eqref{3.7} manifests the essence of the underlay SR system.
For pure CDMA system $(K = 0)$, we can quantify the supportable CDMA load as follows.
Let $q/\sigma^{2}$ and $\beta^\ast$ be the receive SNR and target SINR for CDMA users, both are the parameters related to system design and target Quality-of-Services (QoS).
From the equation, $\frac{q}{\alpha^\ast_{\rm{MF}} q+ \sigma^{2}} = \beta^\ast$, we can obtain the supportable CDMA load as
\begin{align} \label{3.8}
\alpha^\ast_{\rm{MF}} = \frac{1}{\beta^\ast} - \frac{1}{q/\sigma^2}.
\end{align}
When the CDMA load decreases from $\alpha^\ast_{\rm{MF}}$ to $\alpha$, to maintain the same target SINR $\beta^\ast$, there is interference margin tolerable by CDMA system, which we denote as $T_{{\rm{MF}}}$. Based on the equation $\frac{q}{\alpha q+T_{{\rm{MF}}} + \sigma^2} = \beta^\ast$, $T_{{\rm{MF}}}$ can be derived as
\begin{align} \label{3.9}
T_{{\rm{MF}}} = (\alpha^\ast_{\rm{MF}} - \alpha) q,
\end{align}
which is the maximal interference level that can be introduced by OFDMA system.

\subsection{CDMA SINR with MMSE Receiver}

When the CDMA system adoptes MMSE receiver, the receiver for user $u$, ${\bm{d}}_u$, can be denoted as \cite{IEEEhowto: Goldsmith}
\begin{align} \label{3.14}
	{{\bm{d}}_u} = {\left( {\sum\nolimits_{u=1}^{U} {q {\bm{e}}_i {{{\bm{e}}_i^H}}}  + {\bm{\Sigma}} + {\sigma ^2}{\bm{I}}} \right)^{ - 1}} {\bm{e}}_u {q},
\end{align}
where ${\bm{e}}_i = {\bm{\Lambda}}_{i} {\tilde{{\bm{s}}}}_i$.
Then, the output of the receiver is $\hat a_u^{({\rm{MMSE}})} = {\bm{d}}_u^H\tilde {\bm{r}}$, whose SINR can be evaluated as
\begin{align} \label{3.15}
    \gamma _u^{a({\rm{MMSE}})} = q{\bm{e}}_u^H{\left( {\sum\nolimits_{i = 1}^U {q{{\bm{e}}_i}{\bm{e}}_i^H}  + {\bm{\Sigma}}  + {\sigma ^2}{\bm{I}}} \right)^{ - 1}}{{\bm{e}}_u}.
\end{align}

Similar to MF, the SINR of MMSE receiver cannot be evaluated without knowing the specific spreading codes and CSI of all users.
It has been proven that the asymptotic CDMA SINR with MMSE receiver without OFDMA users converges to a deterministic value, under the help of random matrix theory.
To derive the asymptotic CDMA SINR in the OFDMA/CDMA SR system, we first recall the following theorem, whose derivation considers the interference among CDMA users only.
\begin{theorem}[Theorem 6.10 of \cite{IEEEhowto: Couillet}]
    With {\bf{AS1}}, the SINR of user $u$ without OFDMA users converges to $x_u(-\sigma^2)$ in probability, where $x_u(z)$, $z\in \mathbb{C} \backslash \mathbb{R}^+$ is the unique Stieltjes transform that satisfies\footnote[3]{The proof is provided as a special case in \cite{IEEEhowto: Wagner}.}
\begin{align}
    {x_u(z)} = \frac{1}{N}\sum\limits_{n = 1}^N {\frac{{q{{\left| {{\lambda _{u,n}}} \right|}^2}}}{{\frac{1}{N}\sum\nolimits_{i = 1}^U {\frac{q}{{1 + {x_i(z)}}}{{\left| {{\lambda _{i,n}}} \right|}^2}}  - {z}}}}.
\end{align}
\end{theorem}
Considering the interference introduced by OFDMA transmission, the asymptotic SINR of CDMA users is given in the following corollary.
\begin{corollary} \label{corollary: MMSE-SINR}
With {\bf{AS1}}, the SINR of user $u$, $ \gamma _u^{a({\rm{MMSE}})}$, in the SR system converges in probability to $x_u(-\sigma^2)$, where $x_u(z)$, $z\in \mathbb{C} \backslash \mathbb{R}^+$ is the unique Stieltjes transform that satisfies
\begin{align} \label{3.16}
	 x_u(z)  = \frac{1}{N}\sum\limits_{n=1}^N {\frac{{{q}{{\left| {{\lambda _{u,n}}} \right|}^2}}}{{\frac{1}{N}\sum\nolimits_{i=1}^U {\frac{{{q}}}{{1 + {x_i(z)}}}{{\left| {{\lambda _{i,n}}} \right|}^2}}  + \sigma _n^2 - {z}}}}.
\end{align}
\end{corollary}
\begin{proof}
As $N, U \to \infty$,
\begin{align} \label{c3.1}
    \gamma _u^{a({\rm{MMSE}})}
     \mathop \to \limits^{a.s.}
     \frac{q}{N}{\rm{Tr}} \Bigg\{ {\bm{\Lambda}}_u{{\bm{\Lambda}} ^H_u} {\Bigg( { \sum\limits_{i = 1}^U q {{{\bm{e}}_i}{\bm{e}}_i^H}  + {\bm{\Sigma}}  + {\sigma ^2}{\bm{I}}} \Bigg)^{ - 1}} \Bigg\}.
\end{align}
As ${\bm{\Lambda}}_u{{\bm{\Lambda}}_u^H}$ is Hermitian matrix with uniformly bounded spectrum norm, and ${\bm{\Sigma}} = {\rm{diag}}(\sigma_1^2,\dots, \sigma_N^2)$, the trace in \eqref{c3.1} can be further derived as
\begin{align} \label{c3.2}
    & \frac{q}{N}{\rm{Tr}} \Bigg\{ {\bm{\Lambda}}_u {{\bm{\Lambda}} ^H_u} {\Bigg( { \sum\limits_{i = 1}^U q {{{\bm{e}}_i}{\bm{e}}_i^H}  + {\bm{\Sigma}}  - {z}{\bm{I}}} \Bigg)^{ - 1}} \Bigg\} \notag \\
    & \mathop \to \limits^{a.s.}
     \frac{q}{N}{\rm{Tr}} \Bigg\{ {{\bm{\Lambda}} _u}{\bm{\Lambda}} _u^H{\Bigg( {\sum\limits_{i = 1}^U {\frac{q}{{1 + {x_i(z)}}}{{\bm{\Lambda}} _i} {\bm{\Lambda}} _i^H}  + {\bm{\Sigma}}  - {z}I} \Bigg)^{ - 1}} \Bigg\},
\end{align}
where $x_i(z)$ is the unique functional solution of
\begin{align}
    x_i(z) = \frac{q}{N}{\rm{Tr}} \Bigg\{ {{\bm{\Lambda}} _u}{\bm{\Lambda}} _u^H{\Bigg( {\sum\limits_{j = 1}^U {\frac{q}{{1 + {x_j(z)}}}{{\bm{\Lambda}} _i} {\bm{\Lambda}} _i^H}  + {\bm{\Sigma}}  - {z}{\bm{I}}} \Bigg)^{ - 1}} \Bigg\},
\end{align}
such that all $\{x_i(z)\}_{i=1,\dots,U}$ are Stieltjes transforms of non-negtive finite measure on $\mathbb R^+$.
Since ${\bm{\Lambda \Lambda}}^H$ and $\bm{\Sigma}$ are diagonal matrices, the trace in \eqref{c3.2} can be readily derived as
\begin{align}\label{c3.3}
	 \sum\limits_{n = 1}^N {\frac{{{{\left| {{\lambda _{u,n}}} \right|}^2}}}{{\frac{1}{N}\sum\nolimits_{i = 1}^U {\frac{q}{{1 + {x_i(z)}}}{{\left| {{\lambda _{i,n}}} \right|}^2}}  + \sigma _n^2 - {z}}}}.
\end{align}
Thus we can conclude that $\gamma^{a\rm{(MMSE)}}_u $ converges to $x_u(-\sigma^2)$ in probability, where $x_u(z)$ is the unique Stieltjes transform that satisfies \eqref{c3.3}.
Then, the corollary is proven.
\end{proof}

The asymptotic SINR of CDMA users with MMSE receiver under colored noise was also investigated in \cite{IEEEhowto: Viswanath}, where the asymptotic SINR of CDMA users depends on the distribution of the covariance of noise on each dimension.
The distribution is pre-assumed in the context of that work, since the colored noise modeled the interference from neighboring cells with CDMA transmission, which can be estimated but cannot be controlled.
In this work, however, the colored noise reflects the interference introduced by the OFDMA transmission, whose transmission strategy, i.e. resource allocation, will be controlled by the OFDMA system.

Similar to the MF scenario, the result obtained under selective-fading channel shown in \eqref{3.16} can be extended to flat-fading and AWGN channel, respectively. Under the flat-fading channel condition, \eqref{3.16} can be simplified as
\begin{align} \label{3.17}
    {x_u(z)} = \frac{1}{N}\sum\limits_{n=1}^N {\frac{{{q}{{\left| {{\lambda _u}} \right|}^2}}}{{\frac{1}{N}\sum\nolimits_{i=1}^U {\frac{{{q}}}{{1 + {x_i(z)}}}{{\left| {{\lambda _i}} \right|}^2}}  + \sum\nolimits_{k=1}^K {{p_{k,n}}{g_{k,n}}}  + {\sigma ^2}}}},
\end{align}
while under the AWGN channel condition,  \eqref{3.17} can be further simplified as
\begin{align} \label{3.19}
    {x_u(z)} = \frac{1}{N}\sum\limits_{n=1}^N {\frac{{{q}}}{{\frac{1}{N}\sum\nolimits_{i=1}^U {\frac{{{q}}}{{1 + {x_i(z)}}}}  + \sum\nolimits_{k=1}^K {{p_{k,n}}g_{k,n}}  + {\sigma ^2}}}}.
\end{align}

Solving \eqref{3.16} to find the SINR for each user involves solving $U$ coupled non-linear equations.  Moreover, it also needs to know the explicit CSI of all the CDMA users. To overcome these difficulties, we provide the following corollary.

\begin{corollary}
With {\bf{AS1}} and {\bf{AS2}}, the SINR of all the CDMA users with MMSE receivers converges to $\gamma^{a(\rm{MMSE})}$, which is the unique solution of
\begin{align}\label{3.20}
	{x} = {{\mathbb E}_{n}}\left[ {\frac{{{q}}}{{{\frac{{{\alpha q}}}{{1 + {x}}}}  + \sigma _n^2 + {\sigma ^2}}}} \right],
\end{align}
where ${\mathbb E}_{n}[\cdot]$ denotes taking the arithmetic mean on $\{\sigma_n^2\}_{n=1,\dots,N}$.
\end{corollary}
\begin{proof}
The proof is given in Appendix B.
\end{proof}
 Similar to the MF case, without the OFDMA sharing ($K=0$), the supportable CDMA load when MMSE is adopted can be derived as
\begin{align}
	\alpha^\ast_{\rm{MMSE}} = \left( \frac{1}{\beta^\ast} - \frac{1}{q/\sigma^2} \right) \left( 1+\beta^\ast \right).
\end{align}
Similar to the MF scenario, when $\alpha^\ast_{\rm{MMSE}}$ decreases to $\alpha$, there is interference margin provided by the CDMA system.
Let $t^{{\rm{MMSE}}}_n$ denote the interference margin on subcarrier $n$, based on equation ${\beta^\ast} = {{\mathbb E}_{n}}[ {\frac{{{q}}}{{{\frac{{{\alpha q}}}{{1 + \beta^\ast}}}  + t^{{\rm{MMSE}}}_n + {\sigma ^2}}}} ]$, we can see that $\{t^{{\rm{MMSE}}}_n\}$ can be different across subcarriers.
Therefore, it is difficult to derive the specific interference margin for all subcarriers.
In the following section, we will solve this problem by reinforcing the protection to the CDMA users, by which the average interference margin, $T_{{\rm{MMSE}}}$, is needed to be derived only.

\subsection{OFDMA SINR}

At the OFDMA receiver, the output of the FFT processor can be shown in \eqref{2.7}.
The second component is the desired signal, while the first and last components are interference plus noise.
Then, the SINR of OFDMA user $k, \forall k=1,\dots,K$, on subcarrier $n, \forall n=1,\dots,N$, becomes
\begin{align} \label{3.22}
   \gamma _{k,n}^b = \frac{{{p_{k,n}}{g_{k,n}}}}{{\frac{1}{N}\sum\nolimits_{u=1}^U {{q}{{\left| {{\lambda _{u,n}}} \right|}^2}}  + {\sigma ^2}}}.
\end{align}
Under {\bf{AS1}} and {\bf{AS2}}, the SINR in \eqref{3.22} converges to
\begin{align} \label{3.23}
	\gamma _{k,n}^b = \frac{{{p_{k,n}}{g_{k,n}}}}{{\alpha q + {\sigma ^2}}}.
\end{align}
Based on OFDMA definition, we make
\begin{align} \label{3.24}
{p_{k,n}}\left\{ \begin{array}{l}
 > 0,{\rm{ if ~subcarrier~}}n~{\rm{is ~allocated ~to ~user ~}}k\\
 = 0,{\rm{ otherwise}}
\end{array} \right.
\end{align}
to guarantee the subcarriers being exclusively assigned among users. With \eqref{3.23}, the overall uplink throughput achieved by OFDMA system can be written as
\begin{align} \label{3.25}
C = \sum\limits_{k=1}^K {\sum\limits_{n=1}^{N} {{{\log }_2}\left( {1 + \gamma _{k,n}^b} \right)} }.
\end{align}

In next section, we will formulate and solve the OFDMA resource allocation problem with the protection to the CDMA users.

\section{Optimal OFDMA Resource Allocation in SR System}

\subsection{CDMA Protection}

With the receive SNR $q/\sigma^2$ and system load $\alpha$ of CDMA system that are available, the OFDMA system can predict the interference margin based on the SINR analysis in the previous section, where $q/\sigma^2$ and $\alpha$ are both system parameter.
Meanwhile, no information about the OFDMA system is required by the CDMA system.
Thus, this SR system requires least information exchange, which largely alleviates the signalling burden.

To protect the CDMA services, we have to guarantee its SINR no less than the target SINR, $\beta^\ast$, i.e.,
\begin{align}\label{4.1,2}
	\gamma ^{a({\rm{MF}})} \ge {\beta^\ast} ~~ {\rm{or}} ~~ \gamma^{a({\rm{MMSE}})} \ge {\beta^\ast}.
\end{align}
Since the interference margin when CDMA adopts MF receiver has been given by \eqref{3.9}, to meet $\gamma ^{a({\rm{MF}})} \ge {\beta^\ast}$, we can restrict the total interference introduced by OFDMA system according to
\begin{align} \label{4.3}
	\overline{{\sigma _n^2}} \le T_{\rm{MF}}.
\end{align}

To meet the requirement for the CDMA system with MMSE receiver is not straightforward, since $ \gamma ^{a({\rm{MMSE}})} $ is self-contained and dependent on the value of $\{\sigma_n^2\}_{n=1,\dots,N}$ that is to be determined by the OFDMA resource allocation. To avoid solving \eqref{3.20} directly, the following proposition is provided.

\emph{Proposition} 1: Suppose \eqref{3.20}
has the unique solution $\gamma$. For any given $\beta^\ast$, $\gamma \ge \beta^\ast$ if and only if
\begin{align} \label{4.6}
	{{\mathbb{E}}_{n}}\left[ {\frac{q}{{ \frac{\alpha q}{{1 + \beta^\ast }} + \sigma _n^2 + {\sigma ^2}}}} \right] \ge \beta^\ast.
\end{align}

\emph{Proof}: Letting $f\left( x \right) = {{\mathbb{E}}_{n}} \left[ {\frac{{q/x}}{{ \frac{\alpha q}{{1 + x}} + \sigma _n^2 + {\sigma ^2}}}} \right]$, which is a continuous and strictly decreasing function. Since $\gamma$ is the unique solution of \eqref{3.20}, $f(\gamma)=1$, we have $\gamma \ge \beta^\ast$ is equivalent to $f(\beta^\ast)\ge 1$ according to the monotonicity of $f(x)$.

\subsection{Problem Formulation}
Based on above analysis, the OFDMA resource allocation problem in the SR system can be formulated as (P.1).
\begin{align} \label{4.4}
\left( {{\rm{P}}.{\rm{1}}} \right) ~& {\max _{{\bm{P}}}}\sum\limits_{k = 1}^K {\sum\limits_{n = 1}^{{N}} {{{\log }_2}\Big( {1 + \frac{{{p_{k,n}}{g_{k,n}}}}{{\alpha q + {\sigma ^2}}}} \Big)} }    \notag \\
& s.t.\left\{ {\begin{array}{*{20}{l}}
\eqref{4.3} {~\rm{or}} ~\eqref{4.6} \\
{\sum\nolimits_{n = 1}^{N} {{p_{k,n}}}  \le {{\bar P}_k},\forall k=1,\dots,K}
\end{array}} \right.
\end{align}
with $p_{k,n}$ being defined as in \eqref{3.24}, and $\bm{P}$ is the $K\times N$ power allocation matrix with entries of $\{p_{k,n}\}_{k=1,...,K, n=1,...,N}$. $\bar{P}_k$ is the maximum transmission power for each user $k$.

\subsection{OFDMA Resource Allocation with CDMA MF Receiver}

When CDMA adopts MF receiver, (P.1) should be solved with first constraint being \eqref{4.3}.
Apparently, (P.1) is not a convex problem on $\bm{P}$, which makes the problem difficult to solve; and the exhaustive search becomes prohibitively complex as the numbers of subcarriers and users getting large.
It has been shown that the duality gap vanishes, as the number of subcarriers becomes large, by using dual decomposition method \cite{IEEEhowto: Cioffi}, which is also adopt in the uplink resource allocation \cite{IEEEhowto: Huang}.
Therefore, as a benchmark performance, we also adopt this method in solving (P.1). The details of dual decomposition optimization for solving (P.1) is provided in Appendix C.
During solving the problem, $K+1$ dual variables are involved, where $\delta$ is the one associated with the interference constraint and $\lambda_k, k=1,...,K$, associated with power constraints for each OFDMA user. Subgradient method is adopted to update the dual variables.

\subsection{OFDMA Resource Allocation with CDMA MMSE Receiver}

When CDMA adopts MMSE receiver, we should solve (P.1) with the first constraint being \eqref{4.6}.
However, the problem cannot be decoupled completely w.r.t. OFDMA users and subcarriers; thus the dual decomposition method cannot be applied.
To make the problem tractable, we will reinforce this constraint.

Since the interior function of the expectation in \eqref{4.6} is concave, we have
\begin{align} \label{4.7}
   {\mathbb E}_n \left[{\frac{q}{{\frac{{\alpha q}}{{1 + \beta }} + \sigma _n^2 + {\sigma ^2}}}} \right] \ge \frac{q}{{\frac{{\alpha q}}{{1 + \beta }} + \overline {\sigma _n^2}  + {\sigma ^2}}}.
\end{align}
If we make that
\begin{align} \label{4.8}
    \frac{q}{{\frac{{\alpha q}}{{1 + \beta }} + \overline {\sigma _n^2}  + {\sigma ^2}}} \ge \beta^\ast,
\end{align}
the constraint \eqref{4.6} can be guaranteed. By solving the equation $\frac{q}{{\frac{{\alpha q}}{{1 + \beta }} + {T_{{\rm{MMSE}}}} + {\sigma ^2}}} = {\beta ^ * }$, we can get the interference margin provided by the CDMA system with MMSE receiver as
\begin{align} \label{4.10}
T_{\rm{MMSE}} = \frac{\left( \alpha^\ast_{\rm{MMSE}} - \alpha \right)q}{1+\beta^\ast},
\end{align}
where $\alpha^\ast_{\rm{MMSE}}$ is given by (30).
Thus, to guarantee \eqref{4.8}, we can restrict the interference power introduced by OFDMA system according to
\begin{align} \label{4.9}
    \overline {\sigma _n^2} \le T_{\rm{MMSE}}.
\end{align}
Now, the problem can be solved by the same algorithm as that of MF case, by updating the interference margin to $T_{\rm{MMSE}}$.

\subsection{Remarks}

Besides solving the problem by typical numerical method, we would like to investigate some interesting structures of the solution for (P.1).
First, we consider a situation that the power limitation is dominant for this problem, for either fairly small ${\bar P}_k$ or large interference margin due to light CDMA load. Then interference constraint can be ignored.
\begin{align}
    \left( {{\rm{P}}.{\rm{2}}} \right) ~& {\max _{\bm{P}}}\sum\limits_{k=1}^{K} {\sum\limits_{n=1}^{N} {{{\log }_2}\Big( {1 + \frac{{{p_{k,n}}{g_{k,n}}}}{{\alpha q + {\sigma ^2}}}} \Big)} }    \notag \\
    & s.t. \sum\nolimits_{n = 1}^N {{p_{k,n}}}  \le {\bar P _k}, \forall k=1,\dots,K
\end{align}
The optimal power allocation is
\begin{align}
    {p_{k,n}} = {\left[ {\frac{1}{{{\lambda _k}\ln 2}} - \frac{{\alpha q + {\sigma ^2}}}{{{g_{k,n}}}}} \right]^ + },
\end{align}
which is a typical water-filling solution.

As a counterpart, when the interference constraint is stringent, due to a heavy CDMA load or high power supply of the OFDMA user, the problem becomes
\begin{align}
    \left( {{\rm{P}}.{\rm{3}}} \right) ~& {\max _{\bm{P}}}\sum\limits_{k=1}^{K} {\sum\limits_{n=1}^{N} {{{\log }_2}\Big( {1 + \frac{{{p_{k,n}}{g_{k,n}}}}{{\alpha q + {\sigma ^2}}}} \Big)} }    \notag \\
    & s.t. \sum\nolimits_{n=1}^N \sum\nolimits_{k=1}^{K} {{p_{k,n}}{g_{k,g}}}   \le T
\end{align}
where $T$ represents either $T_{\rm{MF}}$ or $T_{\rm{MMSE}}$. Associating dual variable $\delta$ with the constraint, the optimal $p_{k,n}$ is
\begin{align}
    {p_{k,n}} = \frac{1}{{{g_{k,n}}}}{\left[ {\frac{1}{{\delta \ln 2}} - \left( {\alpha q + {\sigma ^2}} \right)} \right]^ + }.
\end{align}
The above equation shows that OFDMA allocates power in a channel-inverse manner. In this case, the OFDMA achievable throughput has the closed form expressed as
\begin{align}
	{C^\prime} = N{\log _2}\Big( {1 + \frac{T}{{\alpha q + {\sigma ^2}}}} \Big).
\end{align}

\section{Performance Evaluation}

In this section, simulation results are provided to evaluate the performance of the proposed OFDMA/CDMA SR system.
The spreading gain of CDMA system and the FFT size of the OFDMA system are set to be $N = 256$, which is large enough to verify the asymptotic results obtained in this paper.
The power of white Gaussian noises is normalized to $\sigma^2=1$. The number of multipaths $L$ is set as $N/8$, and each of the time-domain delay taps follows the distribution of ${\cal{CN}} (0,1/L)$. The target SINR threshold $\beta^\ast$ for CDMA system is set to be 2dB. The designed receive SNR $q/\sigma^2$ and CDMA load $\alpha$ will be varied in the simulations. There are two users in the OFDMA system, and the maximal transmission SNR for each OFDMA user $\bar{P}_k/\sigma^2$ is assumed to be 30dB.

\begin{figure}
    \centering
    \includegraphics[height=6cm]{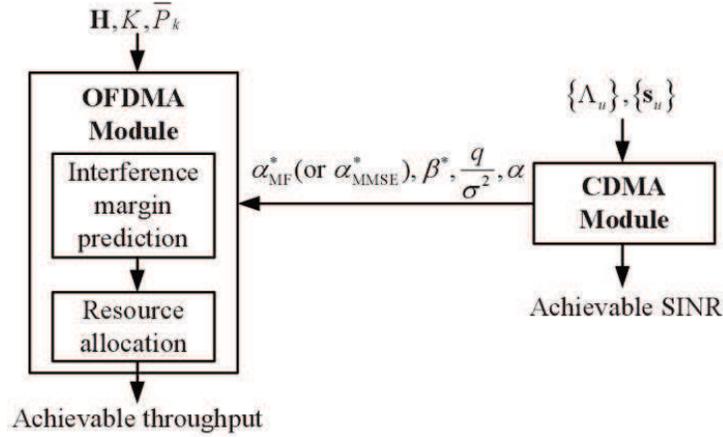}
    \caption{Block diagram of the simulation and information flows.} \label{fig.3}
\end{figure}

The block diagram of the simulation and information flows is illustrated in Fig. \ref{fig.3}. The CDMA system passes its system parameters, $(\alpha^{*}_{\rm{MF}} $ or $\alpha^{*}_{\rm{MMSE}} $, $\beta^{*}, q/\sigma^2, \alpha)$, to the OFDMA system for predicting the interference margin. This is the  only inter-system information flow needed in the SR system.
In order to get the simulated SINRs for CDMA users, the spreading codes for each user are independently and randomly generated, and the interference power introduced by OFDMA system is passed back to the CDMA system. In practice, however, this information flow is not necessary since the interference power can be estimated by the CDMA receiver directly \cite{IEEEhowto: Kansal}, \cite{IEEEhowto: Honig}.

Firstly, we adopt MF as the CDMA receiver to study the OFDMA resource allocation results. For one particular channel realization shown in Fig.\ref{fig.4}~(a) , Fig. \ref{fig.4}~(b) and (c) illustrate the OFDMA power allocation at each subcarrier for light CDMA load scenario and heavy CDMA load scenario, respectively.
It can be observed that for both scenarios, each subcarrier is allocated to the user with the better channel gain.
Moreover, in the light CDMA load scenario, the OFDMA power is allocated in a water-filling way as shown in Fig.~\ref{fig.4}~(b), while in the heavy CDMA load scenario, the OFDMA power is allocated in channel-inverse manner as shown in Fig.~\ref{fig.4}~(c).
These observations are consistent with the discussions in Part E of Section IV.

Fig. \ref{fig.5} and Fig. \ref{fig.8} further validate the convergence of the proposed OFDMA resource allocation algorithm. Fig. \ref{fig.5} (a)-(c) represent light CDMA load case, while Fig. \ref{fig.5} (d)-(f) represent heavy CDMA load case.
The duality gap for both scenarios converges to zero as shown in Fig. \ref{fig.5} (a) and (d).
In Fig. \ref{fig.5} (b) and (c), the converged dual variables $\{\lambda_k\}_{k=1,2}$ are around 0.08, and $\delta$ is close to zero.
This means that the user transmits at its maximum allowable power, and interference to CDMA is less than the interference margin, which are demonstrated in Fig. \ref{fig.8} (b) and (c).
In the heavy load counterpart, the converged $\{\lambda_k\}_{k=1,2}$ approch zeros, while the converged $\delta$ is around 0.02.
This demonstrates that the total interference introduced by the OFDMA system to the CDMA system reaches the interference margin, while the transmission power is less the the maximal value, which can be seen in Fig. \ref{fig.8} (e) and (f).
All these validate the effectiveness of the proposed OFDMA resource allocation algorithm.

\begin{figure}
    \centering
    \includegraphics[height=7cm]{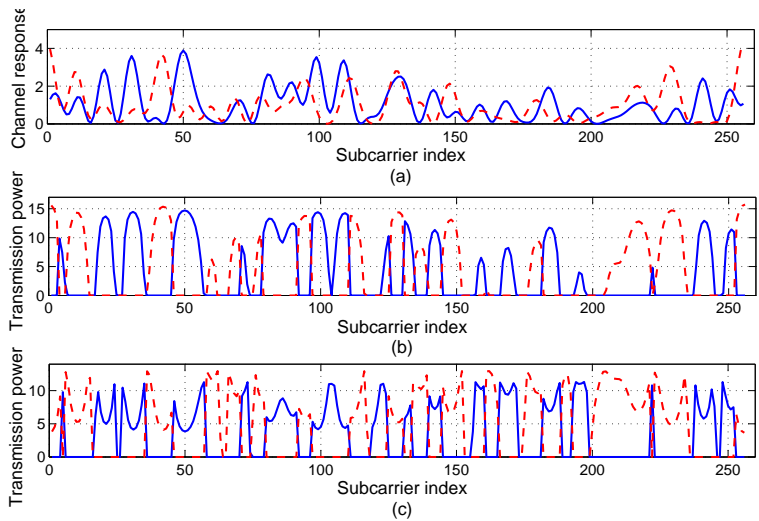}
    \caption{OFDMA resource allocation results: (a) Channel responses of OFDMA users; (b) Subcarrier and power allocation with light CDMA load; (c) Subcarrier and power allocation with heavy CDMA load. The solid line and dashed line represent OFDMA user 1 and OFDMA user 2, respectively.} \label{fig.4}
\end{figure}

\begin{figure}
    \centering
    \includegraphics[height=7cm]{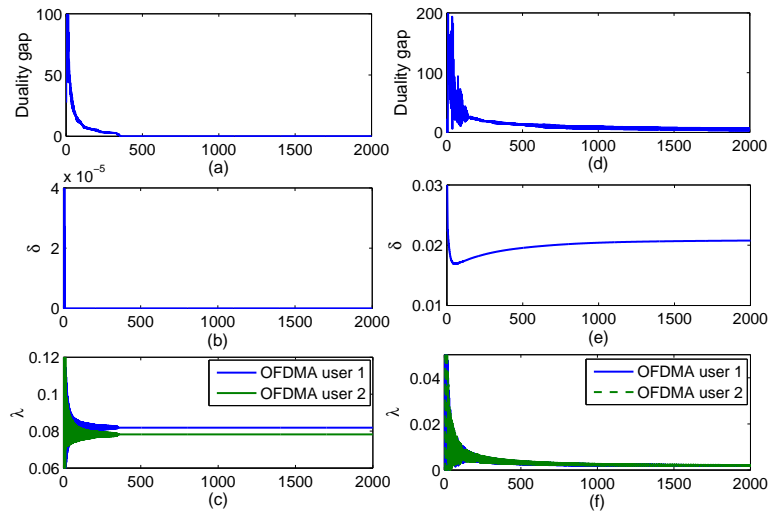}
    \caption{Evolution of duality gap and dual variables. $\delta$ is associated with interference constraint, while $\lambda$ is associated with maximal transmission power constraint. (a)-(c) for light CDMA load; (d)-(f) for heavy CDMA load.} \label{fig.5}
\end{figure}

\begin{figure}
    \centering
    \includegraphics[height=7cm]{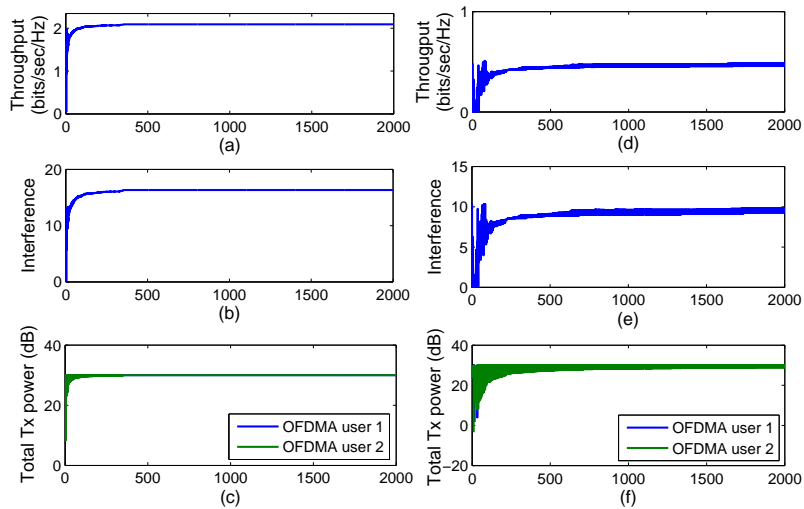}
    \caption{Evolution of OFDMA throughput, interference power to CDMA, and total OFDMA transmission power by each user: (a)-(c) for light CDMA load; (d)-(f) for heavy CDMA load.} \label{fig.8}
\end{figure}

Next, we evaluate the OFDMA achievable throughput by varying the CDMA load, when CDMA adopts different receivers, i.e., MF and MMSE receiver.
From Fig. \ref{fig.10}, we find that OFDMA can achieve much higher throughput when the CDMA system adopts MMSE receiver, as compared with that of MF receiver.
This is because the CDMA system with MMSE receiver can provide higher interference margin than that with MF receiver, which can be seen from \eqref{4.4} and \eqref{4.10}.
Moreover, the feasible range of CDMA load with MMSE receiver is much larger than that with MF receiver.
We also consider two levels of CDMA receive SNR, ${q/}{\sigma^2}=20$dB and ${q/}{\sigma^2}=10$dB to represent high and low CDMA receive power.
When the CDMA load is light, the OFDMA system can achieve high throughput. But when the CDMA load is relative high for both receivers, the 10dB-curve outperforms the 20dB-curve.
This is because the high receive SNR of the CDMA system provides higher interference margin, but also imposes higher interference to OFDMA users.

\begin{figure}
    \centering
    \includegraphics[height=7cm]{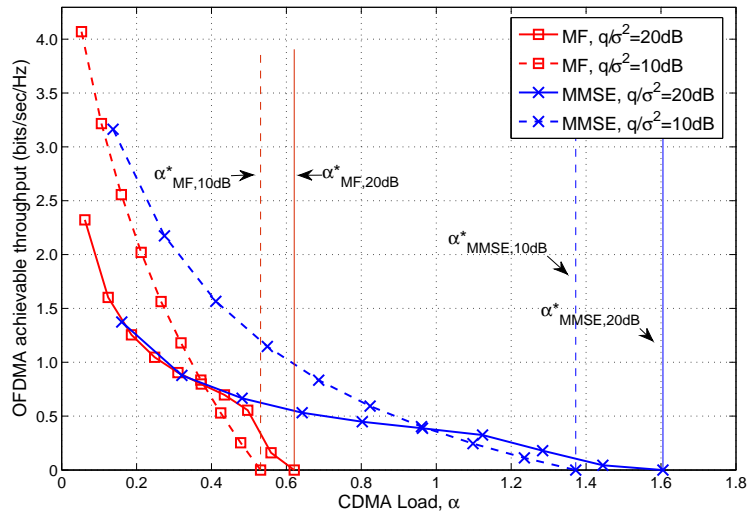}
    \caption{OFDMA achievable throughput vs CDMA load.} \label{fig.10}
\end{figure}

Only looking into the OFDMA throughput is not enough, since the scheme should protect the CDMA services well.
In Fig. \ref{fig.11}, we compare the theoretical and simulated average SINR of the CDMA users by varying the CDMA load, and applying different types of CDMA receivers.
It can be seen that these two values match very well for all the cases.
Furthermore, for both types of CDMA receivers, when the CDMA receive power is high and the CDMA load is light, the achieved CDMA SINR is higher than the target SINR. The reason is, although small $\alpha$ provides high interference margin, it cannot be exploited by the OFDMA system due to the transmit power limit. As $\alpha$ increases, the interference margin becomes less, and eventually can be fully exploited by the OFDMA system. Thus, the achieved SINR for the CDMA system converges to the target value, when the CDMA load becomes heavy.

\begin{figure}
    \centering
    \includegraphics[height=8cm]{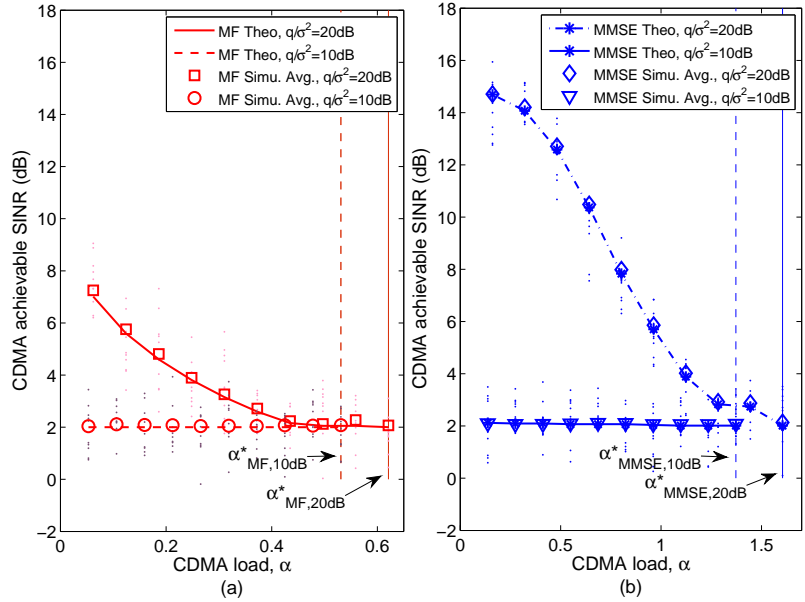}
    \caption{CDMA achievable SINR vs CDMA load: (a) MF receiver; (b) MMSE receiver} \label{fig.11}
\end{figure}

Furthermore, we verify the OFDMA achievable throughput by varying the CDMA receive SNR ${q/}{\sigma^2}$.
Two levels of CDMA load, i.e., $\alpha=0.05$ and $\alpha=0.2$ are considered.
Fig. \ref{fig.12} shows that for a given CDMA load, the OFDMA system achieves higher throughput when the CDMA system adopts the MMSE receiver.
For the same type of CDMA receiver, the OFDMA system achieves higher throughput when the CDMA load is light.
In addition, there is an optimal CDMA receive SNR value that maximizes the OFDMA achievable throughput.
The reason can be stated as follows.
When the CDMA receive SNR is small, the interference margin that can be exploited by OFDMA is small; thus, the OFDMA can only achieve low throughput even though it has extra transmission power.
As the CDMA receive SNR increases, the interference margin increases, thus OFDMA can achieve higher throughput through transmitting at higher power.
When the CDMA receive SNR further increases, although higher interference margin is provided by the CDMA system, this margin cannot be fully exploited by OFDMA system due to the transmission power limit of the OFDMA users.
On the other hand, the interference from the CDMA system to the OFDMA users increases, making the OFDMA achieved throughput deteriorate.

\begin{figure}
    \centering
    \includegraphics[width = 10cm, height=7cm]{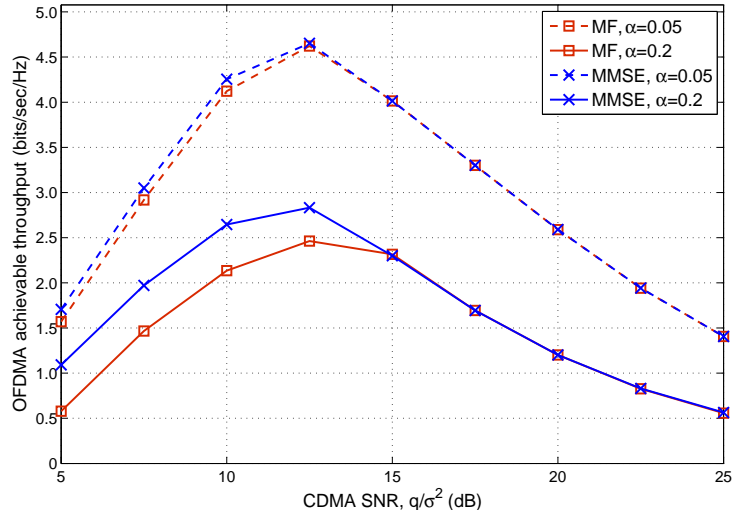}
    \caption{OFDMA achievable throughput vs CDMA receive SNR.} \label{fig.12}
\end{figure}

In the CDMA counterpart shown in Fig. \ref{fig.13}, we can see that the CDMA achieved SINR increases with the receive SNR, and the theoretical and simulated average SINR match very well. The increment of CDMA receive SNR provides higher interference margin to the OFDMA system, but the OFDMA system will not be able exploit it due to its power limit. Eventually, the interference from OFDMA to CDMA will converge to a limit.
Thus, the CDMA achieved SINR will keep increasing as the receive SNR increases.
We also obverse that although the reinforcement of the constraint should have brought along an over-protection to CDMA user, the effect is negligible as the number of subcarriers $N$ is large.

\begin{figure}
    \centering
    \includegraphics[width = 10cm, height=7cm]{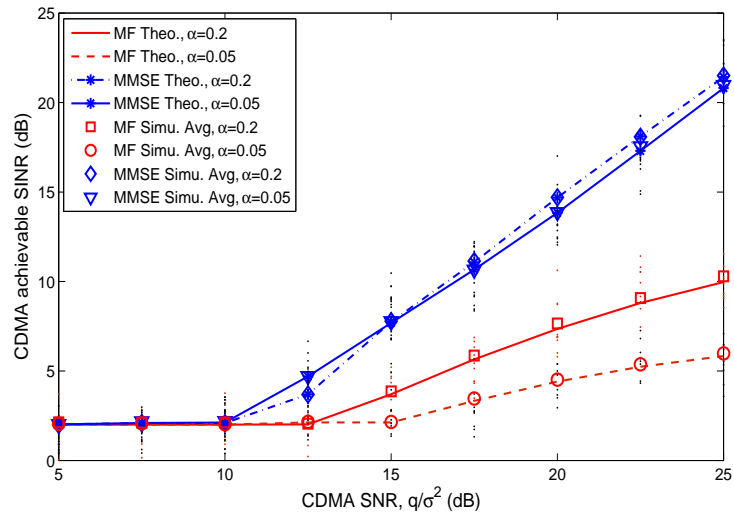}
    \caption{CDMA achievable SINR vs CDMA receive SNR.} \label{fig.13}
\end{figure}

\section{Conclusion}

In this paper, we have proposed an underlay OFDMA/CDMA SR system which allows OFDMA network to operate in the spectrum allocated to the CDMA network. We have quantified the mutual interference between uplink transmissions of the two systems, and derive the asymptotic SINR of the CDMA users. The interference margin that can be tolerated by the CDMA system is derived by the random matrix theory and the large number law.
With the interference margin together with the transmit power constraints, the resource allocation problem of OFDMA system is formulated and solved through dual decomposition method. Our simulation results have verified our theoretical analysis, and validated the effectiveness of the proposed resource allocation algorithm and its capability to protect the legacy CDMA users. The proposed SR system requires the least information flow from the CDMA system to the OFDMA system, and no upgrading of the legacy CDMA system is needed, thus it can be deployed by telecom operators to improve the spectral efficiency of their cellular networks.

\appendices

\section{Proof of Corollary 2} \label{Appendix B}

Because the power of channel response in frequency domain is the same as that of impulse response in time domain, we have
\begin{align}\label{3.8}
	 \frac{1}{N}\sum\nolimits_{n=1}^N {{{\left| {{\lambda _{u,n}}} \right|}^2}}  =  \sum\nolimits_{l=1}^L {{{\left| {{h_{u,l}}} \right|}^2}}.
\end{align}
Considering {\bf{AS1}} and {\bf{AS2}},
\begin{align} \label{3.91}
	\mathop {\lim }\limits_{N \to \infty } \frac{1}{N}\sum\nolimits_{n=1}^N {{{\left| {{\lambda _{u,n}}} \right|}^2}}  = \mathop {\lim }\limits_{L \to \infty } \sum\nolimits_{l=1}^L {{{\left| {{h_{u,l}}} \right|}^2}}.
\end{align}
Based on large number law, and considering that each path has the power of $\frac{1}{L}$ given in {\bf{AS 2}}, the RHS of \eqref{3.91} can be derived as
\begin{align} \label{3.10}
	\mathop {\lim }\limits_{L \to \infty } \sum\nolimits_{l=1}^L {{{\left| {{h_{u,l}}} \right|}^2}}  = \mathop {\lim }\limits_{L \to \infty } L{\mathbb{E}}[|{h_{u,l}}{|^2}] = 1.
\end{align}
Thus, we have $\mathop {\lim }\limits_{L \to \infty } \frac{1}{N}\sum\nolimits_{n=1}^N {{{\left| {{\lambda _{u,n}}} \right|}^2}} =1$.
Applying the large number law to the first term in the denominator of \eqref{3.3} yields
\begin{align} \label{3.11}
& \mathop {\lim }\limits_{N,U \to \infty } \frac{1}{{{N^2}}}\sum\limits_{n=1}^N {\left( {{{\left| {{\lambda _{u,n}}} \right|}^2}{\sum\limits _{i=1,i\not= u}^U}q{{\left| {{\lambda _{i,n}}} \right|}^2}} \right)}  \notag \\
 = & \mathop {\lim }\limits_{N,U \to \infty } \frac{1}{{{N^2}}}{\sum\limits _{i=1, i\not=u}^U}\left( {q\sum\limits_{n=1}^N {{{\left| {{\lambda _{u,n}}} \right|}^2}{{\left| {{\lambda _{i,n}}} \right|}^2}} } \right) \notag \\
 = & \mathop {\lim }\limits_{N, U \to \infty } \frac{q}{N}{\sum\limits _{i=1,i\not=u}^U}\left( {{\mathbb{E}}\left[ {{{\left| {{\lambda _{u,n}}} \right|}^2}{{\left| {{\lambda _{i,n}}} \right|}^2}} \right]} \right)
 = \alpha q.
\end{align}
The last equation holds because of $\frac{U}{N}\to\alpha$. Note that the OFDMA resource allocation is NOT based on the channel condition of CDMA system, since the instantaneous CSI of CDMA is not available to OFDMA; thus it is reasonable to treat them as independent variables. Then, the second term in the denominator of (10) is
\begin{align} \label{3.12}
	\mathop {\lim }\limits_{N \to \infty } \frac{1}{N}\sum\limits_{n=1}^N {\left( {{{\left| {{\lambda _{u,n}}} \right|}^2}\sigma _n^2} \right)}  = {\mathbb E}\left[ {{{\left| {{\lambda _{u,n}}} \right|}^2}\sigma _n^2} \right]
	= {\mathbb E}\left[ {\sigma _n^2} \right].
\end{align}
Substituting \eqref{3.10}-\eqref{3.12} into (10) yields \eqref{3.7}. Thus, the corollary is proven.
\hfill $\qed$

\section{Proof of Corollary 4} \label{Appendix C}

Substituting $z$ with $-\sigma^2$ and applying large number law to \eqref{3.16}, we have
\begin{align} \label{3.21}
{x_u} &= \mathop {\lim }\limits_{U, N \to \infty, \frac{U}{N}\to\alpha} \frac{1}{N}\sum\limits_{n=1}^N {\frac{{{q}{{\left| {{\lambda _{u,n}}} \right|}^2}}}{{\frac{1}{N}\sum\nolimits_{i=1}^U {\frac{{{q}}}{{1 + {x_i}}}{{\left| {{\lambda _{i,n}}} \right|}^2}}  + \sigma _n^2 + {\sigma ^2}}}} \notag \\
 & = \mathop {\lim }\limits_{U, N \to \infty, \frac{U}{N}\to \alpha } \frac{1}{N}\sum\limits_{n=1}^N {\frac{{{q}{{\left| {{\lambda _{u,n}}} \right|}^2}}}{{ {\mathbb E} \left[ {\frac{{{\alpha q}}}{{1 + {x}}}} \right] + \sigma _n^2 + {\sigma ^2}}}} \notag \\
 & = {{\mathbb E}_{n}}\left[ {\frac{{{q}}}{{ {\mathbb E} \left[ {\frac{{{\alpha q}}}{{1 + {x}}}} \right] + \sigma _n^2 + {\sigma ^2}}}} \right] .
\end{align}
Here we can see the SINR among CDMA users is uniform. Thus, we can conclude the SINR of any user is the solutions of \eqref{3.20}. Here, the corollary is proven. \hfill $\qed$

\section{Dual Decomposition in Solving (P.1)} \label{Appendix A}

By denoting ${r_{k,n}} = {\log _2}( {1 + \frac{p_{k,n} g_{k,n}} {{\alpha q + {\sigma ^2}}}} )$, the Lagrangian of (P.1) becomes
\begin{align}\label{A.1}
    L ( {\bm{P}} ,\delta , {\bm{\lambda}} ) = & \sum\limits_{k = 1}^K {\sum\limits_{n = 1}^N {{r_{k,n}}} }  - \delta \left( {\frac{1}{N}\sum\limits_{n = 1}^N {\sum\limits_{k = 1}^K {{p_{k,n}}{g_{k,n}}} }  - T} \right) \notag \\
    & - \sum\limits_{k = 1}^K {{\lambda _k}\left( {\sum\limits_{n = 1}^N {{p_{k,n}}}  - {{\bar P}_k}} \right)},
\end{align}
where $\delta$ and $\{\lambda_k\}_{k=1,...,K}$,  are the Lagrangian dual variables for interference and power constraints, respectively. The Lagrangian dual function is thus
\begin{align}\label{A.2}
   f ( \delta , {\bm{\lambda}} ) = \mathop {\max }\limits_{ {\bm{P}} } L ( {\bm{P}} ,\delta ,  {\bm{\lambda}}),
\end{align}
and its Lagrangian dual problem is
\begin{align}\label{A.3}
    \mathop {\min }\limits_{\delta, {\bm{\lambda}}} f\left( {\delta ,{\bm{\lambda}}} \right).
\end{align}
Decomposing \eqref{A.2} into $N$ independent parallel maximisation problems yields
\begin{align}\label{A.4}
   {f_n}(\delta ,\bm{\lambda} ) = \mathop {\max }\limits_{\bm{P}} \left\{ {\sum\limits_{k = 1}^K {\left( {{r_{k,n}} - {\lambda _k}{p_{k,n}} - \delta {p_{k,n}}{g_{k,n}}} \right)} } \right\}.
\end{align}
Thus $f(\delta, {\bm{\lambda}})$ can be written as
\begin{align} \label{A.5}
   f\left( {\left\{ {{\bm{\lambda}}} \right\},\left\{ {{\delta _u}} \right\}} \right)
   = \sum\limits_{n=1}^N {{f_n}\left( {\delta ,\left\{ {\bm{\lambda}} \right\}} \right)}  + \sum\limits_{k=1}^K {{\lambda _k}{{\bar P}_k}}  + \delta T.
\end{align}
Letting the first-order derivative w.r.t. $p_{k,n}$ of the objective function in \eqref{A.4} equal to zero yields
\begin{align} \label{A.6}
    {p_{k,n}} = {\Big[ {\frac{1}{{\left( {{\lambda _k} + \delta {g_{k,n}}} \right)\ln 2}} - \frac{{\alpha q + {\sigma ^2}}}{{{g_{k,n}}}}} \Big]^ + },
\end{align}
where ${[x]^ + } = \max \left( {0,x} \right)$. Substituting \eqref{A.6} into \eqref{A.4}, the subcarrier can be allocated according to
\begin{align}\label{A.7}
    {f_n}\left( {\delta ,{\bm{\lambda}}} \right) = \mathop {\max }\limits_k \left\{ {{r_{k,n}} - {\lambda _k}{p_{k,n}} - \delta {p_{k,n}}{g_{k,n}}} \right\}.
\end{align}
For each subcarrier $n$, traversing $K$ OFDMA users and allocating it to the best OFDMA user. Then, $\delta$ and $\{\lambda_k\}$ will be searched by subgradient method \cite{IEEEhowto: Boyd}. And one of the subgradient is
\begin{align} \label{A.9}
    {{\bm{d}}_k} = \left[ {\begin{array}{*{20}{c}} {T - \frac{1}{N}\sum\nolimits_{n = 1}^N {\sum\nolimits_{k = 1}^K {{p_{k,n}}{g_{k,n}}} } }\\
    {{{\bar P}_1} - \sum\nolimits_{n = 1}^N {{p_{1,n}}} }\\
     \vdots \\
     {{{\bar P}_K} - \sum\nolimits_{n = 1}^N {{p_{K,n}}} }
     \end{array}} \right].
\end{align} \hfill $\qed$

\end{document}